\documentclass[sigconf, 9pt]{acmart}

\renewcommand\footnotetextcopyrightpermission[1]{} 
\usepackage[utf8]{inputenc}
\usepackage{amsmath}
\usepackage{amsfonts}
\usepackage{amsthm}

\usepackage{float}

\usepackage{graphicx}
\graphicspath{{Figures/}}
\usepackage{epstopdf}

\usepackage{subcaption} 
\usepackage{epsfig}
\usepackage{url}
\usepackage{comment}
\usepackage{graphics}
\usepackage{soul} 
\usepackage{multirow}
\usepackage[export]{adjustbox}

\usepackage{todonotes} 

\usepackage{xfrac}

\usepackage{graphicx, array, blindtext} 

\usepackage{graphicx}

\usepackage{wrapfig}

\usepackage{algorithm}  
\usepackage{algorithmic}

\usepackage{psfrag,wrapfig}

\usepackage{enumerate}
\usepackage{paralist} 
\usepackage{subfiles} 

\usepackage{array}
\usepackage{ragged2e}
\newcolumntype{P}[1]{>{\RaggedRight\hspace{0pt}}p{#1}}

\usepackage{multirow}

\newcommand{\eg}{{\it e.g.,}\xspace}
\newcommand{\viz}{{\it viz.,}\xspace}

\newcommand{\etal}{{\it et~al.}\xspace}
\newcommand{\ie}{{\it i.e.,}\xspace}
\newcommand{\etc}{{\it etc.}}
\newcommand{\ci}{{\it (i) }}
\newcommand{\cii}{{\it (ii) }}

\newcommand{\ca}{{\it (a) }}
\newcommand{\cb}{{\it (b) }}

\newcommand{\note}{{\bf Note: }}

\theoremstyle{acmdefinition}
\newtheorem{definition}{Definition}

 \newcommand{\pnamemcex}{HYDRA-C\xspace}

\newcommand{\prefname}{HYDRA\xspace}

\newcommand{\prefnameglobal}{GLOBAL-TMax\xspace}
\newcommand{\prefnamepartition}{HYDRA-TMax\xspace}

\title{Period Adaptation for Continuous Security Monitoring in Multicore Real-Time Systems}

\renewcommand\footnotetextcopyrightpermission[1]{} 
\author{Monowar Hasan\textsuperscript{*}, Sibin Mohan\textsuperscript{*}, Rodolfo Pellizzoni\textsuperscript{\textdagger}, Rakesh B. Bobba\textsuperscript{\textdaggerdbl}}
\affiliation{\textsuperscript{*}\{mhasan11, sibin\}@illinois.edu, \textsuperscript{\textdagger}rodolfo.pellizzoni@uwaterloo.ca, \textsuperscript{\textdaggerdbl}rakesh.bobba@oregonstate.edu\\} 
\email{}

\begin{document}



\begin{abstract}
We propose a design-time
framework (named \pnamemcex) for integrating security tasks into partitioned\footnote{In \textit{partitioned scheduling} (a widely accepted multicore scheduling scheme), tasks are statically partitioned onto identical cores (\ie runtime migration across cores is not permitted)~\cite{parti_see}.} real-time systems (RTS) running on multicore platforms. Our goal is to \textit{opportunistically} execute security monitoring mechanisms in a `continuous' manner -- \ie as often as possible, across cores, to ensure that security tasks run with as few interruptions as possible. Our framework  will allow designers to integrate security mechanisms without perturbing existing real-time (RT) task properties or execution order. We demonstrate the framework using a proof-of-concept implementation with intrusion detection mechanisms as security tasks. We develop and use both, \ca a custom intrusion detection system (IDS), as well as \cb Tripwire -- an open source data integrity checking tool. These are implemented on a realistic rover platform designed using an ARM multicore chip. We compare the performance of \pnamemcex with a state-of-the-art RT security integration approach for multicore-based RTS and find that our method can, on average, detect intrusions 19.05\% faster without impacting the performance of RT tasks.
\end{abstract}

\maketitle




\section{Introduction}
\label{sec:intro}



Limited resources in terms of processing power, memory, energy, \etc~coupled with the fact that security was not considered a design priority has led to the deployment of a large number of real-time systems (RTS) that include little to no security mechanisms. Hence, retrofitting such legacy RTS with general-purpose security solutions is a challenging problem since any perturbation of the real-time (RT) constraints
(runtimes, periods, task execution orders, deadlines, \etc) could be detrimental
to the correct and safe operation of RTS. 
Besides, security mechanisms need to be designed in such a way that an adversary can not easily evade them. Successful attacks/intrusions into
RTS are often aimed at impacting the safety guarantees of such systems, as evidenced by 
recent intrusions (\eg attacks on control systems~\cite{stuxnet, Ukraine16}, automobiles~\cite{ris_rts_1, checkoway2011comprehensive}, medical devices~\cite{security_medical}, \etc~to name but a few).
Systems with RT properties pose unique security challenges -- these systems are required to meet stringent timing requirements along with strong safety requirements. 
Limited resources (\ie
computational power, storage, energy, \etc) prevent security mechanisms that have been primarily developed for general purpose
systems from being effective for safety-critical RTS. 

In this paper we aim to improve the security posture of 
RTS through integration of security tasks while ensuring that the existing RT tasks are not affected by such integration.  The security tasks considered could be carrying out any one of protection, detection or response-based operations, depending on the system requirements. For instance, a sensor measurement correlation task may be added for detecting sensor manipulation or a change detection task (or other intrusion detection programs) may be added to detect changes/intrusions into the system. In Table \ref{tab:mc_se_ex_table} we present some examples of security tasks that can be integrated into legacy systems (this is by no stretch meant to be an exhaustive list). Note that the addition of any security mechanisms (such as IDS, encryption/authentication, behavior-based monitoring, \etc) may require modification of the system or the RT task parameters as was the case in prior work~\cite{xie2007improving,lin2009static,lesi2017network,lesi2017security, securecore,securecore_memory,securecore_syscal,xie2005dynamic,sibin_RT_security_journal}. 

Further, to provide the best protection, \textit{security tasks may need to be executed as often as possible}. 
If the interval between consecutive checking events is too large then an attacker may remain undetected and cause harm to the system between two invocations of the security task. In contrast, if the security tasks are executed very frequently, it may impact the schedulability of the RT (and other security) tasks. The challenge is then to determining the right \textit{periods} (\ie minimum inter-invocation time) for the security tasks~\cite{sibin_deeply}.

\begin{table}
\caption{Example of Security Tasks}
\label{tab:mc_se_ex_table}
\centering
\hspace*{-0.5em}
\begin{tabular}{P{3.65cm}||P{4.2cm}}
\hline 
\bfseries Security Task & \bfseries Approach/Tools\\
\hline\hline
File-system checking & Tripwire~\cite{tripwire}, AIDE~\cite{aide}, \etc \\
Network packet monitoring & Bro~\cite{bro}, Snort~\cite{snort}, \etc \\
Hardware event monitoring & Statistical analysis based checks~\cite{woo2018early} using performance monitors (\eg \text{perf}~\cite{linux_perf}, \texttt{OProfile}~\cite{oprofile}, \etc)\\ 
Application specific checking & Behavior-based detection (see the related work~\cite{anomaly_detection_survey, securecore,securecore_memory,securecore_syscal}) \\
\hline
\end{tabular}
\end{table}

As a step towards enabling the design of secure RT platforms, opportunistic execution~\cite{mhasan_rtss16,mhasan_date18} has been proposed as a potential way to integrate security mechanisms into legacy RTS --  this allows the execution of security mechanisms as background services without impacting the timing constraints of the RT tasks. Other approaches have been built on this technique for integrating tasks into both legacy and non-legacy systems~\cite{mhasan_ecrts17,xie2007improving,lin2009static,lesi2017network,lesi2017security,securecore,hamad2018prediction}. However, most of that work was focused on single core RTS (that are a majority of such systems in use today). However, \textit{multicore} processors have found increased use in the design of RTS to improve overall performance and energy efficiency~\cite{rt_multicore_lui,mutiprocessor_survey}. 
While the use of such processors \textit{increases} the security problems in RTS (\eg due to parallel execution of critical tasks)~\cite{mhasan_rtiot_sensors19} to our knowledge very few security solutions have been proposed in literature~\cite{mhasan_date18}. In prior work (called \prefname)~\cite{mhasan_date18} researchers have developed a mechanism for integrating security into multicore RTS. However this work uses a partitioned scheduling approach and does not allow runtime migration of security tasks across cores. We show that this results in delayed detection of intrusions\footnote{We discuss this issue further in  Section~\ref{sec:evaluation}.} as the security tasks are not able to execute as frequently.
Our main goal in this paper is \textit{to raise the responsiveness of such security tasks by increasing their frequency of execution}. For instance, consider an intrusion detection system (IDS) -- say one that checks the integrity of file systems. If such a system is interrupted (before it can complete checking the entire system), then an adversary could use that opportunity to intrude into the system and, perhaps, stay resident in the part of the filesystem that has already been checked (assuming that the IDS is carrying out the check in parts). If, on the other hand, the IDS task is able to execute with as few interruptions as possible (\textit{e.g.}, by moving immediately to an empty core when it is interrupted), then there is much higher chance of success and, correspondingly, a much lower chance of a successful adversarial action.

\paragraph*{Our Contributions}

 In this paper, we propose a design-time methodology and a framework named \pnamemcex for partitioned\footnote{Since this is the commonly used multicore scheduling approach for many commercial and open-source OSs (such as OKL4~\cite{okl4}, QNX~\cite{qnx}, RT-Linux~\cite{rt_patch}, \etc) -- mainly due to its simplicity and efficiency~\cite{parti_see,mhasan_date18}.} RTS that \ca leverages semi-partitioned scheduling~\cite{kato2009semi} to enable \emph{continuous execution} of security tasks (\ie execute as frequently as possible) across cores, and \cb does not impact the timing constraints of other, existing, RT tasks. 

\pnamemcex takes advantage of the properties of a multicore platform and allows security tasks to migrate across available cores and execute \textit{opportunistically} (\ie when the RT tasks are not running). This framework extends existing work~\cite{mhasan_date18} and ensures better security (\eg faster detection time) and schedulability (see Section \ref{sec:evaluation}). \pnamemcex is able to do this without violating timing constraints for either the existing RT tasks or the security ones (Section~\ref{sec:se_int_frmakework}). We develop a mathematical model and iterative solution that allows security tasks to execute as frequently as possible while still considering the schedulability constraints of other tasks (Section~\ref{sec:period_selection}). In addition, we also present an implementation on a realistic ARM-based multicore rover platform (running a RT variant of Linux system and realistic security applications). We then perform comparisons with the state-of-the-art~\cite{mhasan_date18} 
(Section~\ref{sec:exp_rover}).
Finally, we carry out a design space exploration using synthetic workloads and study trade-offs for schedulability and security. Our evaluation shows that proposed semi-partitioned approach can achieve better execution frequency for security tasks and consequently quicker intrusion detection ($19.05\%$ faster on average) when compared with both fully-partitioned and global scheduling approaches while providing the same or better schedulability (Section~\ref{sec:exp_synthetic}).

	 \note
	 We do not target our framework towards any specific security mechanism -- our focus is to integrate any designer-provided security solution into a multicore-based RTS.
 In our experiments we used Tripwire~\cite{tripwire} (a data integrity checking tool) as well as our \textit{in-house custom-developed malicious kernel module checker} to demonstrate the feasibility of our approach -- the integration framework proposed in this paper is more broadly applicable to other security mechanisms.

\section{Model and Assumptions}

\subsection{Real-time Tasks and Scheduling Model} \label{sec:rt_model}

Consider a set of $N_R$ RT tasks $\Gamma_R = \lbrace \tau_1, \tau_2, \cdots, \tau_{N_R} \rbrace$, scheduled on a multicore platform with $M$ identical cores $\mathcal{M} = \lbrace \pi_1, \pi_2, \cdots, \tau_{M} \rbrace$. Each RT task $\tau_r$ releases an infinite sequence of task instances, also called \textit{jobs}, and is represented by the tuple $(C_r, T_r, D_r)$ where $C_r$ is the worst-case execution time (WCET), $T_r$ is the minimum inter-arrival time (\eg period) and $D_r$ is the relative deadline. 
The utilization of each task is denoted by $U_r = \frac{C_r}{T_r}$.   
We assume constrained deadlines for RT tasks (\eg $D_r \leq T_r$) and that the task priorities are assigned according to rate-monotonic (RM)~\cite{Liu_n_Layland1973} order (\eg shorter period implies higher priority).

All events in the system happen with the precision of integer clock ticks (\ie processor clock cycles), that is, any time $t$ involved in scheduling is a non-negative integer. In this paper we consider RT tasks that are scheduled using partitioned fixed-priority preemptive scheme~\cite{mutiprocessor_survey} and assigned to the cores using a standard task partitioning algorithm~\cite{parti_see, mutiprocessor_survey}. We further assume that the RT tasks are \textit{schedulable}, \viz the worst-case response time (WCRT), denoted as $\mathcal{R}_r$, is less than deadline (\eg  $\mathcal{R}_r \leq D_r, \forall \tau_r$) and the following necessary and sufficient schedulability condition holds for each RT tasks $\tau_r$ assigned to any given core $\pi_m$~\cite{parti_see}: 
\begin{equation}
\exists t : 0 < t \leq D_r \text{~and~} C_r + \hspace*{-1.5em} \sum\limits_{\tau_i \in hp(\tau_r, \pi_m)} \hspace*{-0.2em}  \left\lceil \frac{t}{T_i} \right\rceil C_i   \leq t,
\end{equation}
 where $hp(\tau_r, \pi_m)$ denotes the set of RT tasks with higher priority than $\tau_r$ assigned to core $\pi_m$.

\subsection{Security Model}

Our focus is on integrating given security mechanisms abstracted as security tasks into a legacy multicore RTS without impacting the RT functionality of the RTS. While we use specific intrusion detection mechanisms (\eg Tripwire) to demonstrate our approach, our approach is somewhat agnostic to the security mechanisms. The security model used and the design of security tasks are orthogonal problems. Since we aim to maximize the frequency of execution of security tasks, security mechanisms whose performance improves with frequency of execution (\eg intrusion monitoring and detection tasks) benefit from our framework.

\section{Security Integration Framework} \label{sec:se_int_frmakework}

We propose to improve the security posture of multicore based RT systems by integrating additional \textit{periodic security tasks} (\eg tasks that are specifically designed for intrusion detection purposes). We highlight that \pnamemcex abstracts security tasks and allows designers to execute \textit{any given techniques}. Our focus here is on integration of a given set of security tasks (\eg intrusion detection mechanisms) in an existing multicore RTS without impacting the RT task parameters (\eg WCET, periods, \etc) or their task execution order. In general, the addition of security mechanisms may increase the execution time of existing tasks~\cite{lin2009static,xie2007improving}
or reduce schedulability~\cite{sibin_RT_security_journal}. 
As we mentioned earlier, our focus is on legacy multicore systems where designers
may not have enough flexibility to modify system parameters
to integrate security mechanisms.
We address this problem by allowing security tasks to execute with a priority lower than all the RT tasks, \ie leverage \emph{opportunistic execution}~\cite{mhasan_rtss16, mhasan_date18}. This way, security tasks will only execute during the slack time (\eg when a core is idle) and the timing requirements of the RT tasks will not be perturbed. However, in contrast to prior work (\prefname)~\cite{mhasan_date18} where the security tasks are statically bound to their respective cores, in this paper we allow security tasks to continuously migrate at runtime (\ie the combined taskset with RT and security tasks follows a semi-partitioned scheduling policy) whenever any core is available (\eg when other  RT or higher-priority security tasks are not running). An illustration of \pnamemcex is presented in Fig.~\ref{fig:se_int_fig} where two RT tasks (represented by blue and green rectangles) are partitioned into two cores and a newly added security task (red rectangle) can move across cores.

\begin{figure}
	\centering
	\includegraphics[scale=0.20]{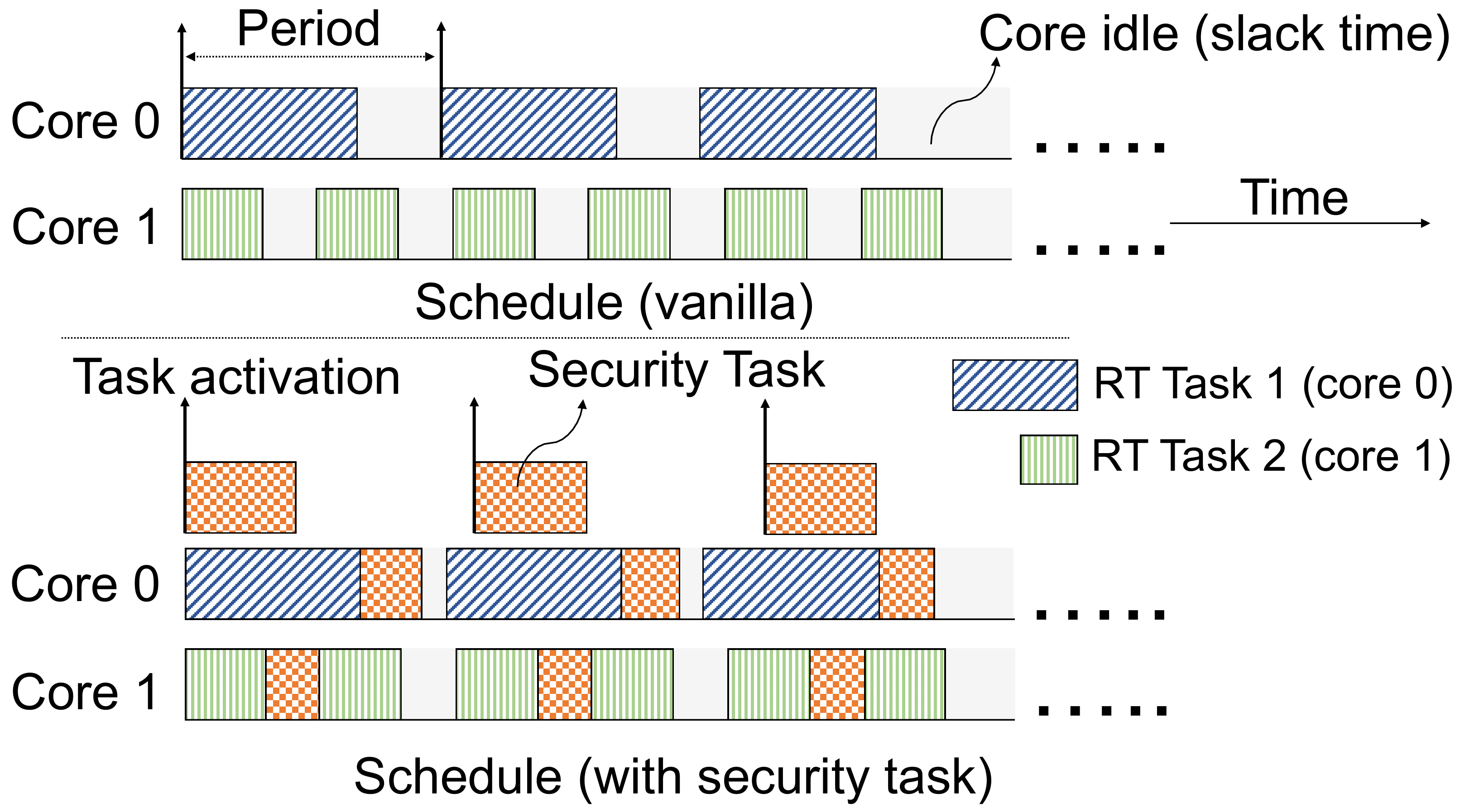}
	\caption{Illustration of our security integration framework for a dual-core platform: two RT tasks (blue and green) are statically assigned to two cores (core 0 and core 1, respectively). We propose to integrate a security task (red) that will execute with lowest priority and can be migrated to ether core (whichever is idle) at runtime.}
	\label{fig:se_int_fig}
\end{figure}

As we shall see in Section \ref{sec:evaluation}, allowing security tasks to execute on any available core will give us the  opportunity to execute security tasks more frequently (\eg with shorter period) and that leads to better responsiveness (faster intrusion detection time). One fundamental question with our security integration approach is to figure out how often to execute security tasks so that the system remains schedulable (\eg WCRT is less than period), and also can execute within a designer provided frequency bound (so that the security checking remains effective). This is different when compared to scheduling traditional RT tasks since the RT task parameters (\eg periods) are often derived from physical system properties and cannot be adjusted due to control/application requirements. We now formally define security tasks.

\subsubsection*{Security Tasks}
Let us include a set of $N_S$ security tasks $\Gamma_S = \lbrace\tau_1, \tau_2, \cdots, \tau_{N_S} \rbrace$ in the system. We adopt the periodic security task model \cite{mhasan_rtss16} and represent each security task by the tuple $(C_s, T_s, T_s^{max})$ where $C_s$ is the WCET, $T_s$ is the (unknown) period (\eg $\tfrac{1}{T_s}$ is the monitoring frequency) and $T_s^{max}$ is a designer provided upper bound of the period -- if the period of the security task is higher than $T_s^{max}$ then the responsiveness is too low and  security checking may not be effective.

We assume that priority of the security tasks are distinct and specified by the designers (\eg derived from specific security requirements). Security tasks have implicit deadlines, \ie they need to finish execution before the next invocation. We also assume that task migration and context switch overhead is negligible compared to WCET of the task. Our goal here is to find a minimum period $T_s \leq T_s^{max}$ (so that the security tasks can execute more frequently)  such that the taskset remains schedulable (\eg $\forall \tau_s \in \Gamma_S$: $\mathcal{R}_s \leq T_s$ where $\mathcal{R}_s$ is the WCRT\footnote{The calculation of WCRT is presented in Section \ref{ref:se_wcrt_cal}.} of $\tau_s$).

\section{Period Selection} \label{sec:period_selection}

The actual periods for the security tasks are not known -- we need to find the periods that ensures schedulability and gives us better monitoring frequency. Mathematically this can be expressed as the following optimization problem: $\underset{T_s, \forall \tau_s \in \Gamma_S}{\operatorname{minimize}} \sum\limits_{\tau_s \in \Gamma_S} T_s$, subject to $\mathcal{R}_s \leq T_s \leq T_s^{max}, \forall \tau_s \in \Gamma_S$. This is a non-trivial problem since the period of $\tau_s$ can be anything in $[\mathcal{R}_s, T_s^{max}]$ and the response time $\mathcal{R}_s$ is variable as it depends on the period of other higher priority security tasks. We first derive the WCRT of the security tasks and use it as a (lower) bound to find the periods. Our WCRT calculation for security tasks is based on the existing iterative analysis for global multicore scheduling~\cite{guan2009new_wcrt_bound,sun2014improving_wcrt2,global_rta_sanjay} and we modify it to account the fact that RT tasks are partitioned.


\subsection{Preliminaries} \label{sec:background}




We start by briefly reviewing the relevant terminology and parameters. We are interested in determining the response time of a job $\tau_s^k$ of task $\tau_s$ (\eg job under analysis) using an iterative method and the response time in each iteration is denoted by $x$. 



\begin{definition}[Busy Period]
The \textit{busy period} of $\tau_s^k$ is the maximal continuous
time interval $[t_1, t_2)$ (until $\tau_s^k$ finishes) where all the cores are executing either higher priority tasks or $\tau_s^k$ itself. 
\end{definition}


\begin{definition}[Interference]
	Given task $\tau_i$, the interference $I_{\tau_s \leftarrow \tau_i}$ caused by $\tau_i$ on $\tau_s^k$ is the number of time units in the busy period when $\tau_i$ executes while $\tau_s^k$ does not.
\end{definition}

Note that the job under analysis $\tau_s^k$ cannot execute if all cores are busy with higher priority tasks; hence, the length of the busy period is at most $\left\lfloor \tfrac{\Omega_s}{M} \right\rfloor +C_s$ by definition, where $\Omega_s$ is the sum of the interference caused by all higher priority tasks on $\tau_s^k$. To compute the value of $I_{\tau_s \leftarrow \tau_i}$, we rely on the concept of \textit{workload}.


\begin{definition}[Workload]
The \textit{workload} $W_i(x)$ of a task $\tau_i$ in a window of length $x$ represents the accumulated execution time of $\tau_i$ within this time interval.
\end{definition}



It remains to compute the workload and corresponding interference for each higher priority task $\tau_i$. We first show how to do so for RT tasks and then for security tasks with higher priority than $\tau_s$.

\subsection{Interference Calculation for RT Tasks} \label{sec:intf_cal}



Since RT tasks are statically partitioned to cores and they have higher priority than any task that is allowed to migrate between cores, the worst-case workload for RT tasks can be trivially obtained based on the same critical instant used for single core fixed-priority scheduling case~\cite{Liu_n_Layland1973}.

\begin{lemma} \label{lemma:rt_workload}
For a given core $\pi_m$, the maximum workload of RT tasks executed on $\pi_m$ in any possible time interval of length $x$ is obtained when all RT tasks are released synchronously at the beginning of the interval.
\end{lemma}
\begin{proof}
Since RT tasks are partitioned and they have higher priorities than security tasks, the schedule of RT tasks executed on $\pi_m$ does not depend on any other task in the system. Now consider any interval $[t, t+x)$ of length $x$. We show that we can obtain an interval $[t', t'+x)$ where all tasks are released at $t'$, such that the workload of RT tasks on $\pi_m$ is higher in $[t', t'+x)$ compared to $[t, t+x)$. 

\textit{First step}: let $t'$ be the earliest time such that $\pi_m$ continuously executes RT tasks in $[t', t)$; if such time does not exist, then let $t' = t$. By definition, $\pi_m$ does not execute RT tasks at time $t'-1$. Also since RT tasks continuously execute in $[t', t)$, the workload of RT tasks in $[t', t'+x)$ cannot be smaller than the workload in $[t, t+x)$.

\textit{Second step}: since $\pi_m$ is idle at $t' - 1$, no job of RT tasks on $\pi_m$ released before $t'$ can contribute to the workload in $[t', t)$. Hence, the workload can be maximized by anticipating the release of each RT task $\tau_r$ so that it corresponds with $t'$. This concludes the proof.
\end{proof}

\begin{figure}
	 \includegraphics[scale=0.28]{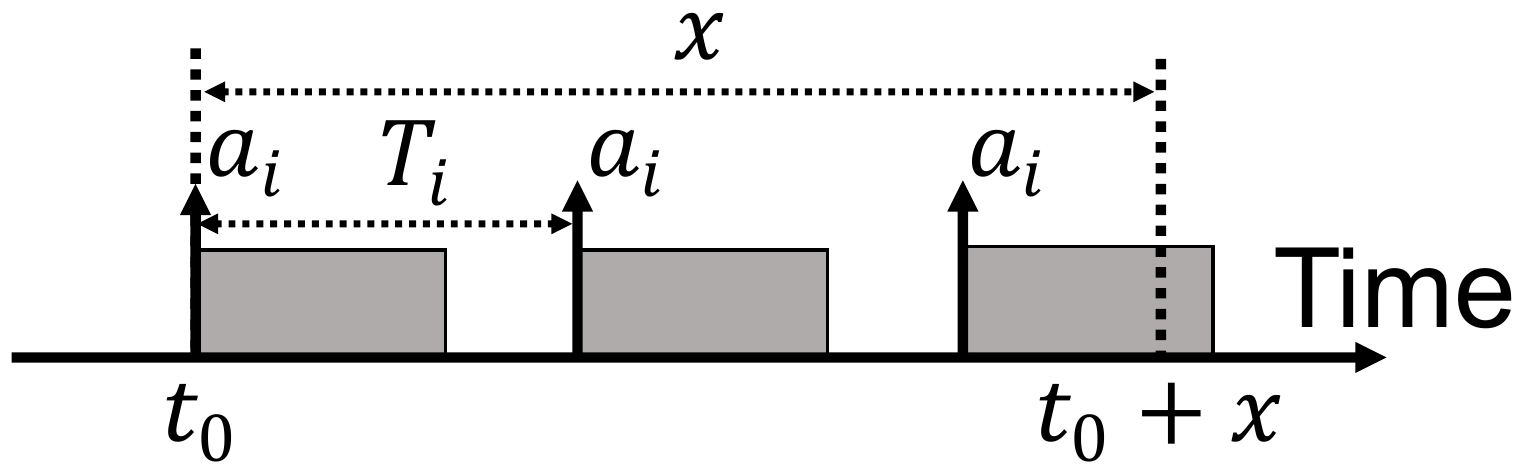}
	\caption{Workload of the RT tasks for a window of size $x$. The arrival time of the task $\tau_{i}$ is denoted by $a_{i}$.}
\label{fig:nc_workload}
\end{figure}

Let $\Gamma_R^{\pi_m} \subseteq \Gamma_R$ denote the set of RT tasks partitioned to core $\pi_m$. Based on Lemma~\ref{lemma:rt_workload}, an 
upper bound to the workload of RT tasks on $\pi_m$ can be obtained by assuming that each RT task $\tau_r$ is released at the beginning of the interval and each job of $\tau_r$ executes as early as possible after being released, as shown in Fig.~\ref{fig:nc_workload}. We thus obtain the workload for RT task $\tau_r$:
\begin{equation} \label{eq:nc_workload}
W_r^{R}(x)=\left\lfloor \frac{x}{T_r} \right\rfloor C_r + \min (x~\mathsf{mod}~T_r, C_r),
\end{equation}
and summing over all RT tasks on $\pi_m$ yields a total workload $\sum\limits_{\tau_i \in \Gamma_R^{\pi_m}} W_i^{R}(x)$. Finally, we notice that by definition the interference caused by a group of tasks executing on the same core $\pi_m$ on $\tau_s$ cannot be greater than $x - C_s + 1$.
Therefore, the maximum interference caused by RT tasks on $\pi_m$ to $\tau_s$ can be bounded as:
\begin{equation} \label{eq:in_rt}
I_{\tau_s \leftarrow  \Gamma_R^{\pi_m}}\Big(x, \hspace*{-0.5em} \sum_{\tau_i \in \Gamma_R^{\pi_m}} \hspace*{-0.5em} W_i^{R}(x)\Big) = \min \left( \sum_{\tau_i \in \Gamma_R^{\pi_m}} \hspace*{-0.2em} W_i^{R}(x), x - C_s + 1 \right).
\end{equation}
The `$+1$' term in the upper bound of the interference (\eg Eq.~(\ref{eq:in_rt})) ensures the convergence of iterative search for the response time (recall from Section~\ref{sec:background} that at each iteration the response time is denoted by $x$) to the correct value~\cite{bertogna2007response}. For example, when the iterative search for the response time is started with $x = C_s$ (\ie $x-C_s = 0$), the search would stop immediately (and outputs an incorrect WCRT) since $\min \left( \sum\limits_{\tau_i \in \Gamma_R^{\pi_m}} \hspace*{-0.2em} W_i^{R}(x), x - C_s \right) = 0$.

\subsection{Interference Calculation for Security Tasks} \label{sec:wcrt_calculation}

We next consider the workload of security tasks with higher priority than $\tau_s$. The workload computation depends on the arrival time of the task relative to the beginning of the busy period, as specified in the following definition.

\begin{definition}[Carry-in]
A task $\tau_i$ is called a \textit{carry-in} task if 
there exists one job of $\tau_i$ that has been released before the beginning of a given time window of length $x$ and executes within the window. 
If no such
job exists, $\tau_i$ is referred to as a \textit{non-carry-in} task. 
\end{definition}

\begin{figure}[!t]
\centering
\includegraphics[width=2.5in]{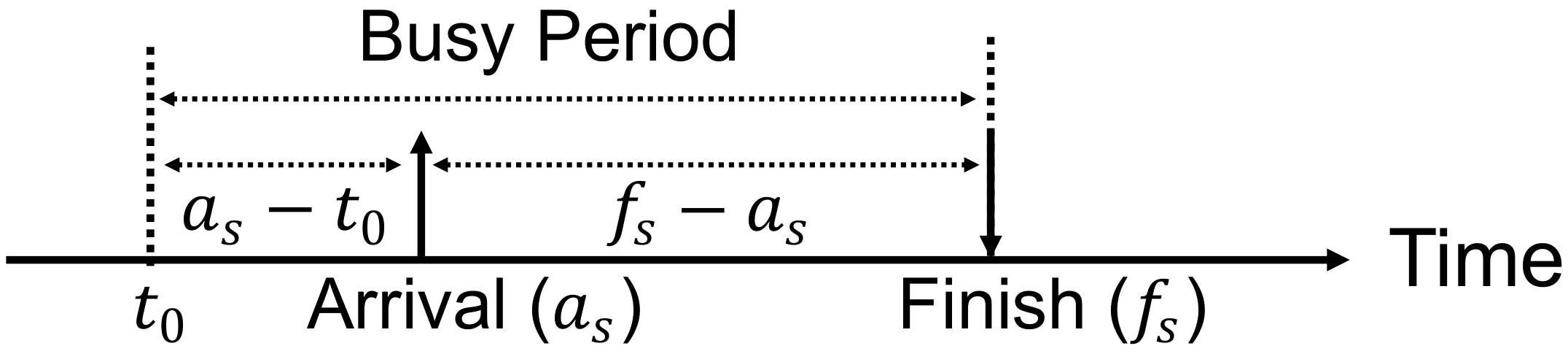}
\caption{Extension of busy period for bounding the number of carry-in higher priority security tasks.}
\label{fig:busy_period_extension}
\end{figure}

Generally (but not always), the workload of a task $\tau_i$ in the busy period is higher if $\tau_i$ is a carry-in task than a non-carry-in task. Hence, it is important to limit the number of higher priority carry-in tasks. To this end, we follow an approach similar to prior research~\cite{guan2009new_wcrt_bound,global_rta_sanjay} and extend the busy period of $\tau_s^k$ from its arrival time (denoted by $a_s$) to an earlier time instance $t_0$ (see Fig.~\ref{fig:busy_period_extension}) such that during any time instance $t \in [t_0, a_s)$ all cores are busy executing tasks with higher priority than $\tau_s$. Note that by definition, this implies that there was at least one free core (\ie not executing higher priority tasks) at time $t_0 - 1$.

\setcounter{theorem}{1}
\begin{lemma}
\label{lemma_ci_se}
At most $M-1$ higher priority tasks can have carry-in at time $t_0$.
\end{lemma} 
\begin{proof}
The maximum number of higher priority tasks that can have carry-in at $t_0$ is $M-1$ since by definition there have to be strictly less than $M$ higher priority tasks active at time $t_0 -1$ (otherwise they will occupy all the cores).
\end{proof}

\begin{figure}
	\centering
	 \includegraphics[scale=0.28]{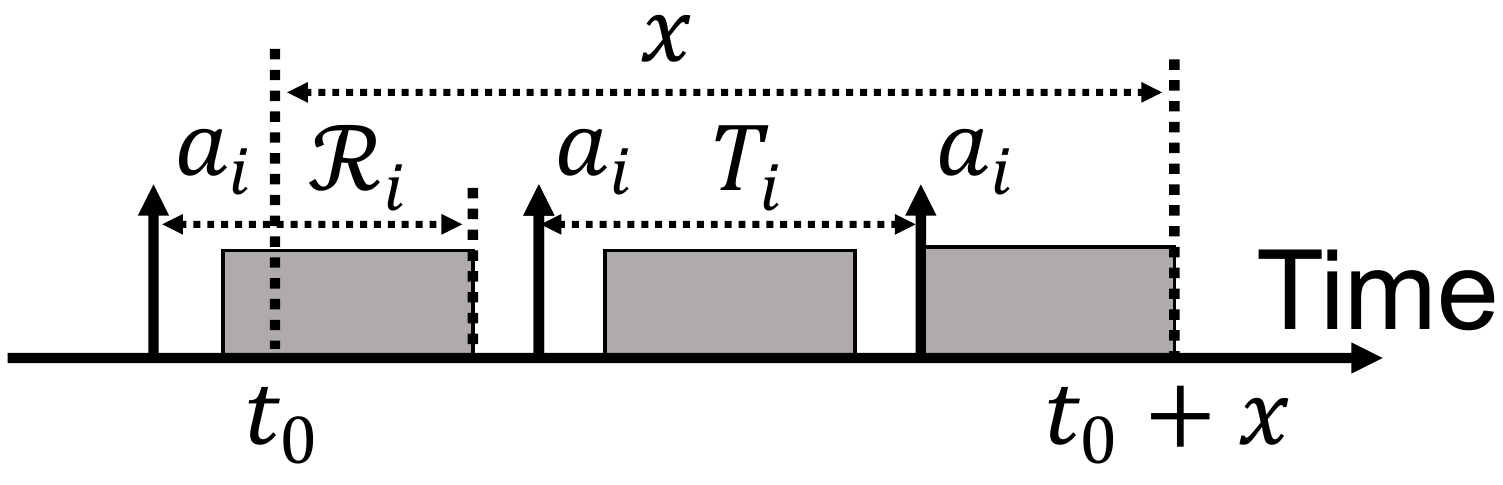}
	\caption{Illustration of carry-in task  for a window of size $x$.}
	\label{fig:ci_workload}
\end{figure}

Since Lemma~\ref{lemma_ci_se} holds for all tasks with higher priority than $\tau_s$, an immediate corollary is that the number of security tasks with carry-in at $t_0$ also cannot be larger than $M-1$. If a security task $\tau_i$ does not have carry-in, its workload is maximized when the task is released at the beginning of the busy interval. Hence, we can calculate the workload bound $W_i^{S_{NC}}(x)$ for the interval $x$ using Eq.~(\ref{eq:nc_workload}), \eg $W_i^{S_{NC}}(x)=\left\lfloor \frac{x}{T_i} \right\rfloor C_i + \min (x~\mathsf{mod}~T_i, C_i)$. Likewise, the workload bound for a carry-in security task $\tau_i$ in an interval of length $x$ starting at $t_0$ is given by (see Fig.~\ref{fig:ci_workload}): 
\begin{equation} \label{eq:ci_workload}
W_i^{S_{CI}}(x) = W_i^{S_{NC}}\left( \max(x - \bar{x}_i, 0) \right) + \min(x, C_i - 1),
\end{equation}
where
$\bar{x}_i = C_i - 1 +  T_i - \mathcal{R}_i$. We can bound the workload of the first carry-in job to $C_i - 1$ because the job must have started executing at the latest at $t_0 - 1$ (given that not all cores are busy).  Finally, using the same argument as in Section~\ref{sec:intf_cal}, the interference of $\tau_i$ can be bounded as follows:
\begin{equation} \label{eq:intf_ci_nc}
I_{\tau_s \leftarrow \tau_i}(x, W_i(x)) = \min \left( W_i(x), x - C_s + 1 \right),
\end{equation}
where $W_i(x)$ is either $W^{S_{NC}}_i(x)$ or $W^{S_{CI}}_i(x)$. Notice that the WCRT and periods of security task in the carry-in workload function 
(see Eq.~(\ref{eq:ci_workload})) 
is actually an unknown parameter. However, we follow an iterative scheme that allows us to calculate the period and WCRT of all higher priority security tasks before we calculate the interference for task $\tau_s$ (refer to Section \ref{sec:alg} for details).





\subsection{Response Time Analysis} \label{ref:se_wcrt_cal}

Let $hp_S(\tau_s$) denote the set of security tasks with a higher priority  than $\tau_s$. Note that we do not know which (at most) $M-1$ security tasks in $hp_S(\tau_s)$ have carry-in. In order to derive the WCRT of $\tau_s$, let us define $\mathcal{Z}_{\tau_s} \subset \Gamma \times \Gamma$ as the set of all partitions of $hp_S(\tau_s)$ into two subsets $\Gamma_s^{NC}$ and $\Gamma_s^{CI}$ (\eg the non overlapping set of carry-in and non-carry-in tasks) such that: 
\begin{equation*}
\Gamma_s^{NC} \cap \Gamma_s^{CI} = \emptyset,
\Gamma_s^{NC} \cup \Gamma_s^{CI} = hp_S(\tau_s),
\text{~and~} 
|\Gamma_s^{CI}| \leq M-1,
\end{equation*}
\eg there are at most $M-1$ carry-in tasks.


For a given carry-in and non-carry-in set (\eg $\Gamma_s^{NC}$ and $\Gamma_s^{CI}$), we can calculate the total interference experienced by $\tau_s$ as follows:
\begin{eqnarray}
\Omega_s(x, \Gamma_s^{NC}, \Gamma_s^{CI})  = \sum_{\pi_m \in \mathcal{M}} I_{\tau_s \leftarrow \Gamma_R^{\pi_m}} \Big(x, \sum_{\tau_i \in \Gamma_R^{\pi_m}} W_i^{R}(x) \Big)~+ \nonumber \\ 
\sum_{\tau_i \in \Gamma_s^{NC}} I_{\tau_s \leftarrow \tau_i}\Big(x, W_i^{S_{NC}}(x) \Big)~ + 
\sum_{\tau_i \in \Gamma_s^{CI}}  I_{\tau_s \leftarrow \tau_i}\Big(x, W_i^{S_{CI}}(x)\Big).
\end{eqnarray}

For a given $\Gamma_s^{NC}, \Gamma_s^{CI}$ sets response time $\mathcal{R}_{s | {(\Gamma_s^{NC}, \Gamma_s^{CI}})}$ will be the minimal solution of the following iteration\footnote{Note that the worst-case is when the job arrives at $t_0$ (\ie $a_s = t_0$).}~\cite{guan2009new_wcrt_bound}:
\begin{equation}
x = \left\lfloor \frac{\Omega_s(x, \Gamma_s^{NC}, \Gamma_s^{CI})}{M}\right\rfloor  + C_s.
\end{equation}
We can solve this using an iterative fixed-point search with the initial condition $x^{(0)}=C_s$. The search terminates if there exists a solution (\ie $x = x^{(k)} = x^{(k-1)}$ for some iteration $k$) or when $x^{(k)} > T_s^{max}$ for any iteration $k$ since $\tau_s$ becomes trivially unschedulable for WCRT greater than $T_s^{max}$. Finally we can calculate the WCRT of $\tau_s$ as follows: 
\begin{equation}
\mathcal{R}_s = \max\limits_{\left( \Gamma_s^{NC}, \Gamma_s^{CI} \right) \in \mathcal{Z}_{\tau_s} } \mathcal{R}_{s | {(\Gamma_s^{NC}, \Gamma_s^{CI}})} .
\end{equation}

\subsection{Algorithm} \label{sec:alg}

The security task $\tau_s$ remains schedulable with any period $T_s \in  [\mathcal{R}_s, T_s^{max}]$. However as mentioned earlier, the calculation of $\mathcal{R}_s$ requires us to know the period and response time of other high priority tasks $\tau_h \in hp_S(\tau_s)$. Also if we arbitrarily set $T_s = \mathcal{R}_s$ (since this allows us to execute security tasks more frequently) it may negatively affect the schedulability of other tasks that are at a lower priority than $\tau_s$ because of a high degree of interference from $\tau_s$. Hence, we developed an iterative algorithm that gives us a trade-off between schedulability and  monitoring frequency.

\renewcommand{\algorithmicforall}{\textbf{for each}}
    \renewcommand\algorithmiccomment[1]{%
 {\it /* {#1} */} %
}
\renewcommand{\algorithmicrequire}{\textbf{Input:}}
    \renewcommand{\algorithmicensure}{\textbf{Output:}}

		\begin{algorithm}[t]
        
			\begin{algorithmic}[1]
			 \begin{footnotesize}
                \REQUIRE Set of real-time and security tasks $\Gamma = \Gamma_R \cup \Gamma_S$
    \ENSURE Periods of the security tasks, $\mathbf{T}$ (if the security tasks are schedulable); $\mathsf{Unschedulable}$	 otherwise	
					\vspace{0.4em}
                   	
                    \STATE Set $T_s := T_s^{max}$ and calculate $\mathcal{R}_s$ for $\forall \tau_s \in \Gamma_S$ 
					\IF {$\exists \tau_s$ such that $\mathcal{R}_s > T_s^{max}$}

                    \STATE \textbf{return} $\mathsf{Unschedulable}$  
                    
                    \ENDIF

					\FORALL{security task $\tau_s \in \Gamma_S$
                    (from higher to lower priority)}

                    \STATE \COMMENT{Find period for which all lower priority tasks are schedulable}
                    \STATE Find minimum $T_s^* \in [\mathcal{R}_s, T_s^{max}]$ using Algorithm \ref{alg:mc_log_search} such that $\forall \tau_j \in lp(\tau_s)$ remains schedulable (\eg $\mathcal{R}_j \leq T_j^{max}$)
                    
                    
                    
                    \STATE Update $\mathcal{R}_j$ for $\forall \tau_j \in lp(\tau_s)$ considering the interference with new period $T_s^*$

					\ENDFOR	
                    \STATE \textbf{return} $\mathbf{T} := [T_s^*]_{\forall \tau_s \in \Gamma_S}$ ~~ \COMMENT{return the periods}

				\end{footnotesize}
				 
			\end{algorithmic}
			\caption{Period Selection}
 \label{alg:mc_period_selection}
		\end{algorithm}

Our proposed solution (refer to Algorithm \ref{alg:mc_period_selection} for a formal description) works as follows. We first fix the period of each security task $T_s^{max}$ and calculate the response time $\mathcal{R}_s$ using the approach presented in Section \ref{sec:wcrt_calculation} (Line 1). If there exists a task $\tau_j$ such that $\mathcal{R}_j > T_j^{max}$ we report the taskset as unschedulable (Line 2) since it is not possible to find a period for the security tasks within the designer provided bounds -- this unschedulability  result will help the designer in modifying the requirements (and perhaps RT tasks' parameters, if possible) accordingly to integrate security tasks for the target system. If the taskset is schedulable with $T_s^{max}$, we then iteratively optimize the periods from higher to lower priority order (Lines 5-9) and return the period (Line 10). To be specific, for each task $\tau_s \in \Gamma_S$ we perform a logarithmic search \cite[Ch. 6]{knuth1997art} (see Algorithm \ref{alg:mc_log_search} for the pseudocode) and find the minimum period $T_s^*$ within the range $[R_s, T_s^{max}]$ such that all low priority tasks (denoted as $lp(\tau_s)$) remain schedulable, \eg $\forall \tau_j \in lp(\tau_s): \mathcal{R}_j \leq T_j^{max}$ (Line 7). Note that since we perform these steps from higher to lower priority order, WCRT and period of all higher priority tasks (\eg $\forall \tau_h \in hp(\tau_s)$) are already known.  We then update the response times of all low priority task $\tau_j \in lp(\tau_s)$ considering the interference from the newly calculated period $T_s^*$ (Line 8) and repeat the search for next security task.

\renewcommand{\algorithmicforall}{\textbf{for each}}
    \renewcommand\algorithmiccomment[1]{%
 {\it /* {#1} */} %
}
\renewcommand{\algorithmicrequire}{\textbf{Input:}}
    \renewcommand{\algorithmicensure}{\textbf{Output:}}

		\begin{algorithm}[t]
        
			\begin{algorithmic}[1]
			 \begin{footnotesize}
                \REQUIRE Set of real-time and security tasks $\Gamma = \Gamma_R \cup \Gamma_S$ 
    \ENSURE A feasible period $T_s^*$ for the security task under analysis (\ie $\tau_s$)
					\vspace{0.4em}
					
					\STATE Define $T_s^l := \mathcal{R}_s, ~~ T_s^r := T_s^{max}, ~~ T_s^c := 0$
					\STATE Set $\widehat{\mathcal{T}}_s := \{T_s^{max}\}$  ~~ \COMMENT{Initialize a variable to store the set of feasible periods}
					\WHILE{$T_s^l <= T_s^r$}
					\STATE Update $T_s^c := \lfloor \frac{T_s^l + T_s^r}{2} \rfloor$
					
					\IF{$\exists \tau_j \in lp(\tau_s)$ such that $\tau_j$ is \textit{not schedulable} with $T_s = T_s^c$}
					\STATE \COMMENT{Increase the period of $\tau_s$ to make the taskset schedulable (\eg by reducing the interference)}
					\STATE Update $T_s^l := T_s^c + 1$
					\ELSE
					\STATE \COMMENT{Taskset is schedulable with $T_s^c$}
					\STATE $\widehat{\mathcal{T}}_s := \widehat{\mathcal{T}}_s \cup \{ T_s^c \}$  ~~ \COMMENT{Add $T_s^c$ to the feasible period list}
					\STATE \COMMENT{Check schedulability with smaller period for next iteration}
					\STATE Update $T_s^r := T_s^c - 1$ 
					\ENDIF

					\ENDWHILE
                   	
                    \STATE Set $T_s^* := \min \left(\widehat{\mathcal{T}}_s \right)$ \COMMENT{Find the minimum period from the set of feasible periods}
                    
                    \STATE \textbf{return} $T_s^*$  ~~ \COMMENT{return the period of $\tau_s$}

				\end{footnotesize}
				 
			\end{algorithmic}
			\caption{Calculation of Minimum Feasible Period for the Security Task $\tau_s$}
 \label{alg:mc_log_search}
		\end{algorithm}

\section{Evaluation} \label{sec:evaluation}


We evaluate \pnamemcex on two fronts: \ci a proof-of-concept implementation on an ARM-based rover platform with security applications -- to demonstrate the viability of our scheme in a realistic setup (Section \ref{sec:exp_rover}); and \cii with synthetically generated workloads for broader design-space exploration (Section \ref{sec:exp_synthetic}). Our implementation code will be made available in a public, open-sourced repository~\cite{mhasan_conex_implementation}.

\subsection{Experiment with an Embedded Platform and Security Applications} \label{sec:exp_rover}

\subsubsection{Platform Overview}




We implemented our ideas on a rover platform
manufactured by Waveshare~\cite{waveshare}.
The rover hardware/peripherals (\eg wheel, motor, servo, sensor, \etc) are controlled by a Raspberry Pi 3 (RPi3) Model B~\cite{rpi3} SBC (single board computer). The RPi3 is equipped with a 1.2 GHz 64-bit quad-core ARM Cortex-A53 CPU on top of  Broadcom BCM2837 SoC (system-on-chip). In our experiments we focus on a dual-core setup (\eg activated only \texttt{core0} and \texttt{core1}) and disabled the other two cores) -- this was done by modifying the boot command file \texttt{/boot/cmdline.txt} and set the flag \texttt{maxcpus}=2. The base hardware unit of the rover is connected with RPi3 using a 40-pin GPIO (general-purpose input/output) header. The rover supports omni-directional movement and can move avoiding obstacles using an infrared sensor (\eg ST188~\cite{ST188}). We also attached a camera (RPi3 camera module) that can capture static images (3280 $\times$ 2464 pixel resolution).  The detailed specifications of the rover hardware (\eg base chassis, adapter, \etc) are available on the vendor website~\cite{waveshare}.


\subsubsection{Experiment Setup and Implementation}


We implemented our security integration scheme in Linux kernel 4.9 and enabled real-time capabilities by applying the PREEMPT\_RT patch~\cite{rt_patch} (version 4.9.80-rt62-v7+). 
In our experiments the rover moved around autonomously and periodically captured images (and stored them in the internal storage).  We assumed implicit deadlines for RT tasks and considered two RT tasks: \ca  a navigation task -- that avoids obstacles (by reading measurements from infrared sensor) and navigates the rover and \cb a camera task that captures and stores still images. We do not make any modifications to the vendor provided control codes (\eg navigation task). In our experiments we used the following parameters $(C_{r}, T_{r})$: $(240, 500)$ ms and $(1120, 5000)$ ms, for navigation and camera tasks, respectively (\ie total RT task utilization was $0.7040)$. We calculated the WCET values using ARM cycle counter registers (CCNT) and set periods in a way that the rover can navigate and capture images without overloading the RPi3 CPU. Since CCNT is not accessible by default, we developed a Linux loadable kernel module and activated the registers so that our measurement scripts can access counter values.

To integrate security into this rover platform, we included two additional security tasks: \ca an open-source security application, Tripwire~\cite{tripwire}, that checks intrusions in the image data-store and  \cb our custom security task that checks current kernel modules (as a preventive measure to detect rootkits) and compares with an expected profile of modules. The WCET of the security tasks were $5342$ ms and $223$ ms, respectively and the maximum periods of security tasks were assumed to be $10000$ ms (\eg total system utilization is at least $0.7040 + 0.5565 = 1.2605$) -- we picked this maximum period value by trial and error so that the taskset became schedulable for demonstration purposes. We used the Linux \texttt{taskset} utility~\cite{linux_taskset} for partitioning tasks to the cores and the tasks were scheduled using Linux native \texttt{sched\_setscheduler()} function. For accuracy of our measurements we disabled all CPU frequency scaling features in the kernel and executed RPi with a constant frequency (\eg 700 MHz -- the default value). The system configurations and tools used in our experiments are summarized in Table \ref{tab:rov_param}.

\begin{table}
\caption{Summary of the Evaluation Platform}
\label{tab:rov_param}
\centering
\begin{tabular}{P{2.7cm}||P{4.9cm}}
\hline 
\bfseries Artifact & \bfseries Configuration/Tools\\
\hline\hline
Platform         & 1.2 GHz 64-bit Broadcom BCM2837  \\
CPU         & ARM Cortex-A53  \\
Memory & 1 Gigabyte  \\ 
Operating System & Debian Linux (Raspbian Stretch Lite)							\\ 
Kernel version   & Linux Kernel 4.9 				\\ 
Real-time patch    & PREEMPT\_RT 4.9.80-rt62-v7+		\\ 
Kernel flags & $\mathtt{CONFIG\_PREEMPT\_RT\_FULL}$ enabled\\		
Boot parameters & $\mathtt{maxcpus}$=2, $\mathtt{force\_turbo}$=1,
$\mathtt{arm\_freq}$=700,
$\mathtt{arm\_freq\_min}$=700 \\
WCET measurement & ARM cycle counter registers \\
Task partition & Linux \texttt{taskset} \\
\hline
\end{tabular}
\end{table}



We compared the performance of our scheme with prior work, \prefname~\cite{mhasan_date18}. In that work, researchers proposed to statically partition the security tasks among the multiple cores -- to our knowledge \prefname is the state-of-the-art mechanism for integrating security in legacy multicore-based RT platforms. The key idea in \prefname was to allocate security tasks using a greedy best-fit strategy: for each task, allocate it to a  core that gives maximum monitoring frequency (\ie shorter period) without violating schedulability constraints of already allocated tasks. 


\subsubsection{Experience and Evaluation}

\begin{figure}
	\centering
	\begin{subfigure}[b]{0.5\linewidth} 
	    \hspace*{-1.5em}
		\centering \includegraphics[scale=0.28]{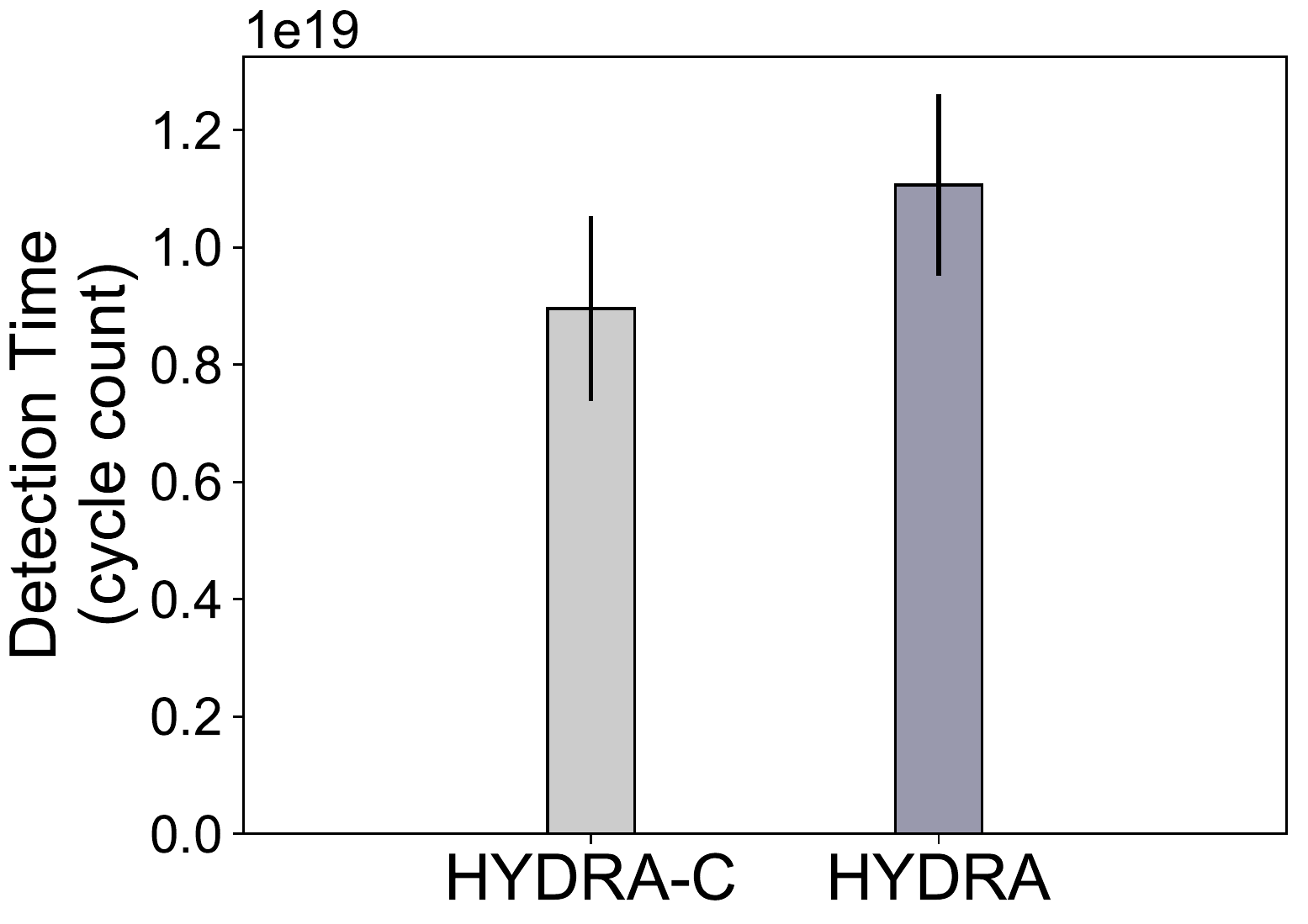}
		\caption{\label{fig:rov_id_time}}
	\end{subfigure}%
	\begin{subfigure}[b]{0.5\linewidth}
	    \hspace*{-0.6em}
		\centering\includegraphics[scale=0.28]{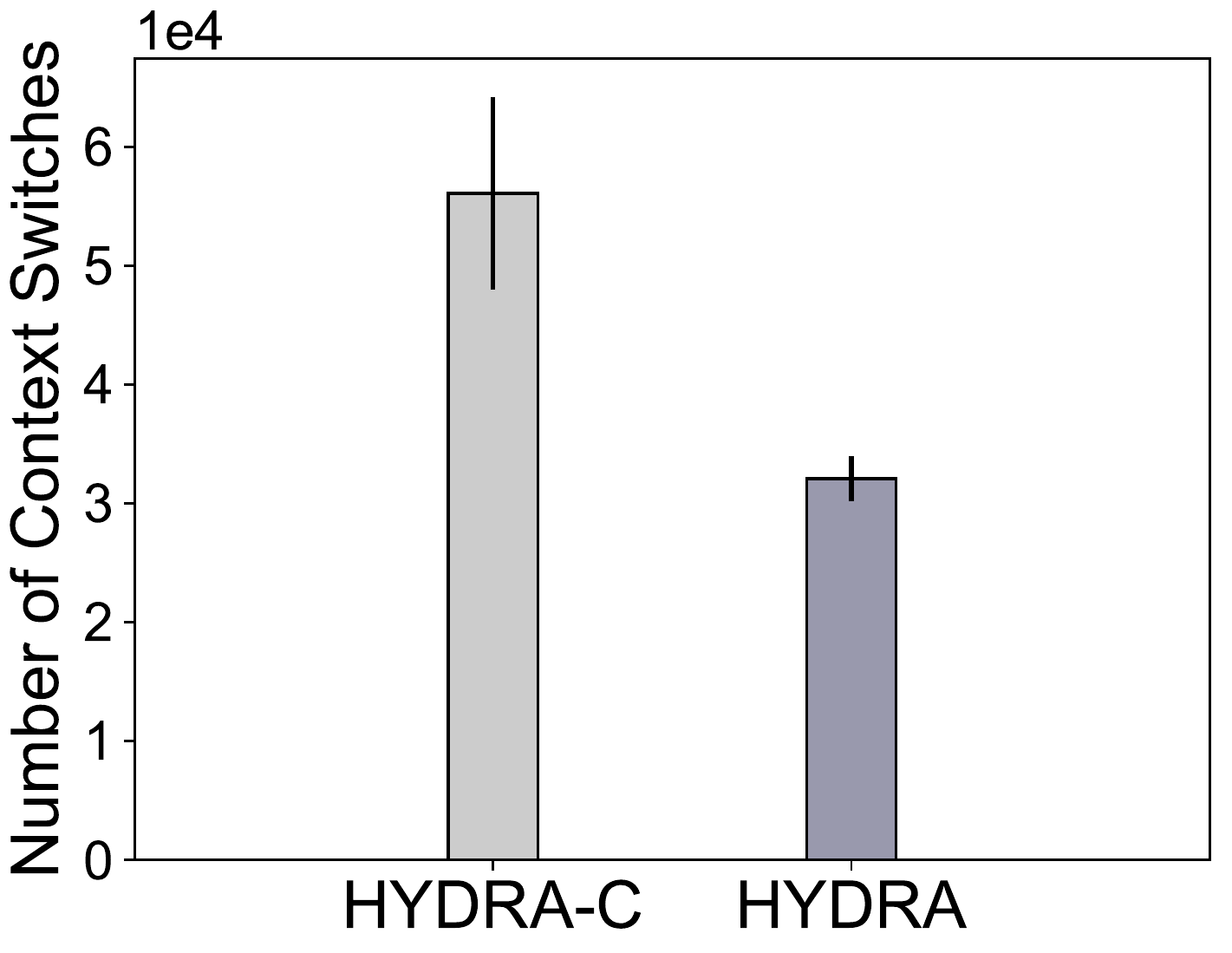}
		\caption{\label{fig:rov_cs_count}}
	\end{subfigure}%
	\caption{Experiments with rover platform:~\textit{(a)}~time (cycle counts) to detect intrusions;~\textit{(b)}~average number of context switches. On average our scheme can detect the intrusions faster without impacting the performance of RT tasks. }
\end{figure}



We observed the performance of \pnamemcex by \textit{analyzing how quickly an intrusion can be detected}. We considered the following two realistic attacks\footnote{\underline{Note}: our focus here is on the integration of any given security mechanisms rather the detection of any particular class of intrusions. Hence we assumed that there were no zero-day attacks and the security tasks were able the detect the corresponding attacks correctly (\ie there were no false-positive/negative errors) -- although the generic framework proposed in this paper allows the designers to accommodate any desired security (\eg intrusion detection/prevention) technique.}: \ci an ARM shellcode~\cite{arm_shellcode} that allows the attacker to modify the contents of the image data-store -- this attack can be detected by Tripwire; \cii a rootkit~\cite{simple_rootkit} that intercepts all the \texttt{read()} system calls -- our custom security task can detect the presence of the malicious kernel module that is used to inject the rootkit. For each of our experimental trials we launched attacks at random points during program execution (\ie from the RT tasks) and used ARM cycle counters to measure the detection time.  In Fig.~\ref{fig:rov_id_time} we show the average time to detect both the intrusions (in terms of cycle counts, collected from $35$ trials) for \pnamemcex and \prefname schemes.  From our experiments we found that, on average, our scheme can detect intrusions  $19.05\%$ faster compared to the \prefname approach (Fig.~\ref{fig:rov_id_time}). Since our scheme allows security tasks to migrate across cores, it provides smaller response time (\eg shorter period) in general and that leads to faster detection times.


We next measured the overhead of our security integration approach in terms of number of context switches (CS). For each of the trials we observed the schedule of the RT and security tasks for $45$ seconds and counted the number of CS using the Linux \texttt{perf} tool~\cite{linux_perf}. In Fig.~\ref{fig:rov_cs_count} we show the number of CS (y-axis in the figure) for \pnamemcex and \prefname schemes (for $35$ trials). As shown in the figure, our approach increases the number of CS (since we permit migration across cores) compared to the other scheme that statically partitions security tasks. From our experiments we found that, on average, our scheme increases CS by 
$1.75$ times.
However, this increased CS overhead \textit{does not impact the deadlines of RT tasks} (since the security tasks always execute with a priority lower than the RT tasks) and thus may be acceptable for many RT applications.

\subsection{Experiment with Synthetic Tasksets} \label{sec:exp_synthetic}

We also conducted experiments with (randomly generated) synthetic workloads for broader design-space exploration. 



\begin{table}
\caption{Simulation Parameters}
\label{tab:mc_ex_param}
\centering
\begin{tabular}{P{5.5cm}||P{2.2cm}}
\hline 
\bfseries Parameter & \bfseries Values\\
\hline\hline
Process cores, $M$ & $\lbrace 2, 4\rbrace$ \\
Number of real-time tasks, $N_R$ & $[3\times M, 10 \times M]$ \\
Number of security tasks, $N_S$ & $[2 \times M, 5 \times M]$ \\
Period distribution (RT and security tasks) & Log-uniform \\
RT task allocation & Best-fit \\
RT task period, $T_r$ & $[10, 1000]$ ms \\
Maximum period for security tasks, $T_s^{max}$ & $[1500, 3000]$ ms\\
Minimum utilization of security tasks & At least $30\%$ of RT tasks \\
Base utilization groups & $10$ \\
Number of taskset in each configuration & $250$\\
\hline
\end{tabular}
\end{table}

\subsubsection{Taskset Generation and Parameters} \label{sec:tc_generation}

In our experiments we used parameters similar to those in related work~\cite{mhasan_date18,sibin_RT_security_journal,mhasan_rtss16,sun2014improving_wcrt2,davis2015global,mn_gp} (see Table \ref{tab:mc_ex_param}). We considered $M \in \lbrace2,4\rbrace$ cores and each taskset instance contained $[3 \times M, 10 \times M]$ RT and $[2 \times M, 5 \times M]$ security tasks. To generate tasksets with an even distribution
of tasks, we
grouped the real-time and security tasksets by base-utilization
from $[(0.01+0.1 i)M, (0.1+0.1i)M]$ where $i \in \mathbb{Z}, 0 \leq i \leq 9$. Each utilization group contained $250$ tasksets (\eg total $10 \times 250 = 2500$ tasksets were tested for each core configuration). We only considered the schedulable tasksets (\eg the condition in Section \ref{sec:rt_model} was satisfied for all RT tasks) -- since tasksets that fail to meet this condition are trivially unschedulable. Task periods were generated according to a log-uniform distribution.  Each RT task had periods between $[10, 1000]$ ms and the maximum periods for security tasks were selected from $[1500 ,3000]$ ms. We assumed that RT tasks were partitioned using a best-fit~\cite{parti_see} strategy and the total utilization of the security tasks was at least $30\%$ of the system utilization. For a given number of tasks and total system utilization, the utilization of individual tasks were generated using Randfixedsum algorithm \cite{randfixedsum}. 




\begin{figure}[!t]
\centering
\includegraphics[scale=0.45]{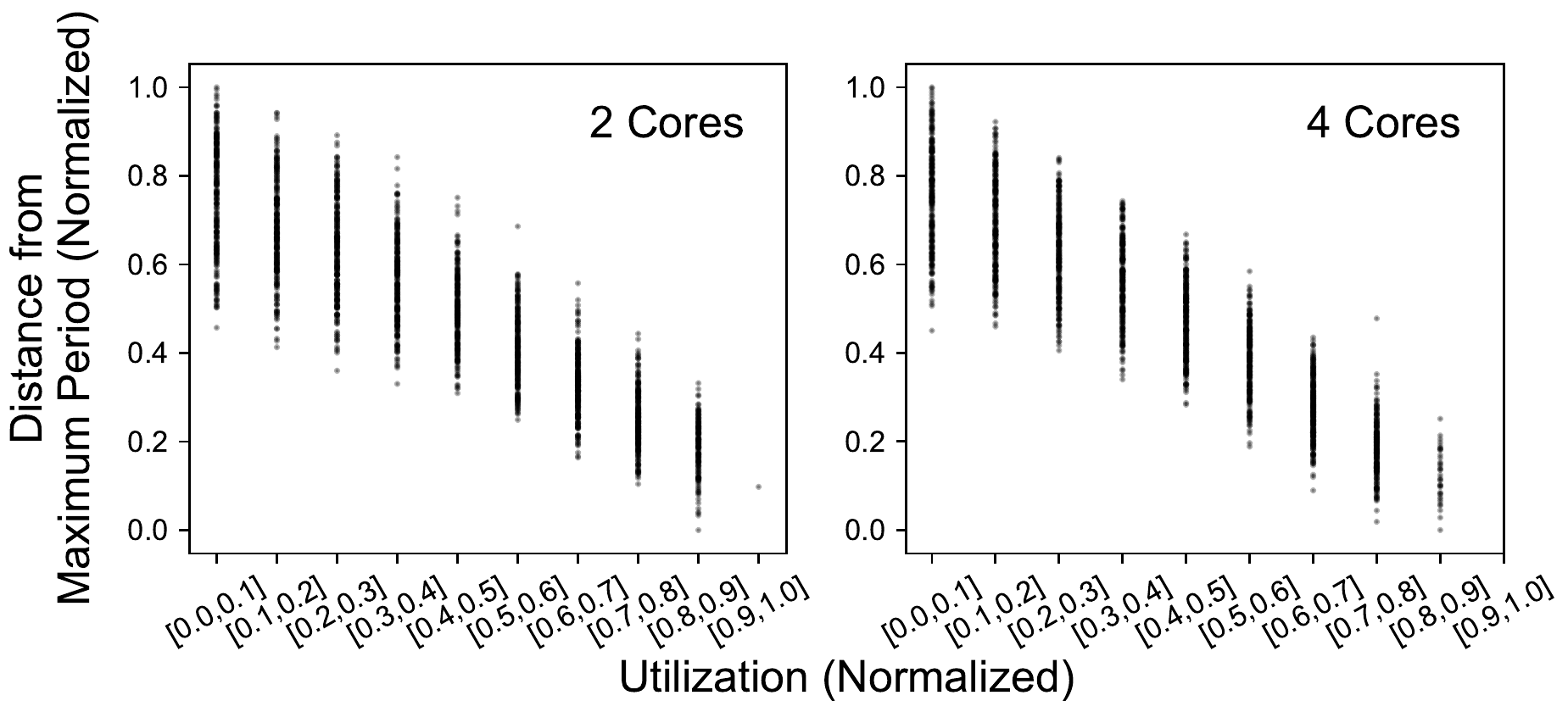}
\caption{Euclidean distance between achievable period and maximum period vectors for different utilizations. 
Larger distance (y-axis in the figure) implies security tasks execute more frequently.}
\label{fig:ecdist_count}
\end{figure}

\begin{figure*}
	\centering
	\captionsetup[subfigure]{oneside,margin={3.2cm,0cm}}
	\begin{subfigure}[]{0.5\linewidth} 
		\centering \includegraphics[scale=0.45]{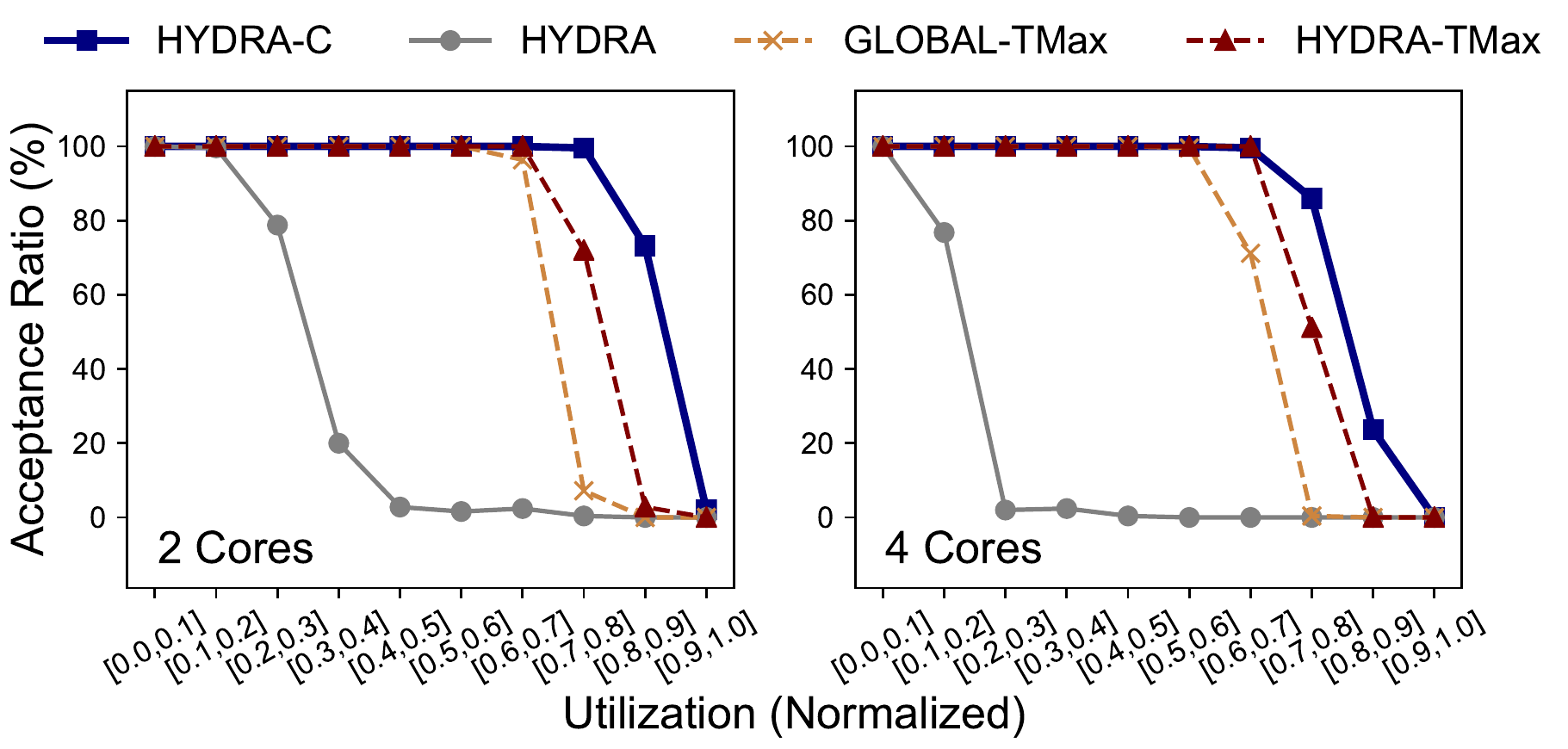}
\caption{} 
\label{fig:sched_count}
	\end{subfigure}%
	\begin{subfigure}[]{0.5\linewidth}
		\centering \includegraphics[scale=0.45]{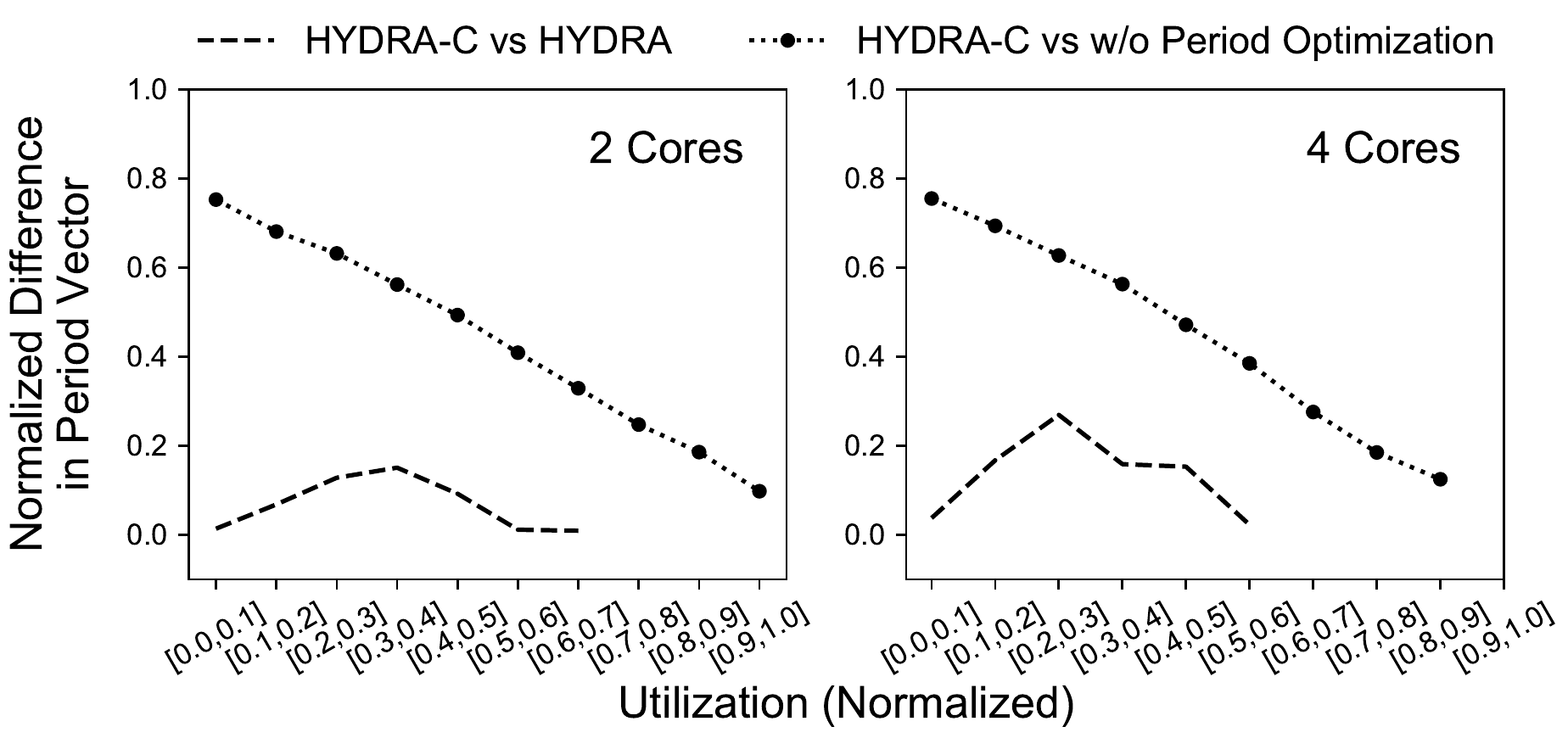}
\caption{}
\label{fig:dist_diff_count}
	\end{subfigure}%
	\caption{Impact on schedulability and security. \textit{(a)} The acceptance ratio vs taskset utilizations for 2 and 4 core platforms: our scheme outperforms \prefname and \prefnameglobal approaches for higher utilizations. \textit{(b)} Difference in period vectors for our approach and reference schemes (\eg \prefname, \prefnameglobal, \prefnamepartition): the non-negative distance (y-axis in the figure) implies that \pnamemcex finds shorter periods than other schemes.} 
	\label{fig:exp_sec_sched_tradeoff}
\end{figure*}

\subsubsection{Impact on Inter-Monitoring Interval}

We first observe how frequently we can execute (schedule) security tasks compared to the designer specified bound (Fig.~\ref{fig:ecdist_count}). The x-axis of Fig.~\ref{fig:ecdist_count} shows the normalized utilization $\tfrac{U}{M}$  where $U$ is the minimum utilization requirement and given as follows: 
$
U = \hspace*{-0.0em}\sum\limits_{\tau_r \in \Gamma_R} \hspace*{-0.0em} \frac{C_r}{T_r}  + \hspace*{-0.0em}\sum\limits_{\tau_s \in \Gamma_S} \hspace*{-0.0em} \frac{C_s}{T_s^{max}}.
$  
The y-axis represents the Euclidean distance  between the calculated period vector $\mathbf{T}^{\boldsymbol{*}} = [T_s^*]_{\forall \tau_s \in \Gamma_S}$ and maximum period vector $\mathbf{T}^{\boldsymbol{\rm max}} = [T_s^{max}]_{\forall \tau_s \in \Gamma_S}$ (normalized to $1$). A higher distance implies that tasks can run more frequently. As we can see from the figure for higher utilizations, the distance reduces (\eg periods are closer to the maximum value) -- this is mainly due to the interference from higher priority (RT and security) tasks. The results from this figure suggest that we can execute security tasks more frequently for low to medium utilization cases. 



\subsubsection{Impact on Schedulability and Security Trade-off}

While in this work we consider a legacy RT system (\ie where RT tasks are partitioned to respective cores), for comparison purposes 
we considered the following two schemes (in addition  to the related work, \prefname, introduced in Section \ref{sec:exp_rover}) that do not consider any period adaptation for security tasks. 
\begin{itemize}
    \item \prefnameglobal: In this scheme all the RT and security tasks are scheduled using a global fixed-priority multicore scheduling scheme~\cite{mutiprocessor_survey}.
    Since our focus here is on schedulability we set $T_s = T_s^{max}, ~\forall \tau_s \in \Gamma_S$ (recall that a taskset can be considered schedulable if the following conditions hold: $\mathcal{R}_r \leq D_r, \forall \tau_r \in \Gamma_R$ and $\mathcal{R}_s \leq T_s^{max}, \forall \tau_s \in \Gamma_S$). This scheme allows us to observe the performance impacts of binding RT tasks to the cores (due to legacy compatibility).
    
    \item \prefnamepartition: This is similar to the \prefname approach introduced in Section \ref{sec:exp_rover} (\ie security tasks were partitioned using
    best-fit allocation as before) but instead of minimizing periods here we set  $T_s = T_s^{max}, \forall \tau_s$. This allows us to observe the trade-offs between schedulability and security in a fully-partitioned system.

    
\end{itemize}


In Fig.~\ref{fig:sched_count} we compare the performance of \pnamemcex with the \prefname, \prefnameglobal and \prefnamepartition strategies in terms of \textit{acceptance ratio} (y-axis in the figure) defined as the number of schedulable tasksets (\eg $\mathcal{R}_s \leq T_s^{max}, \forall \tau_s$) over the generated one and the x-axis 
shows the normalized utilization $\tfrac{U}{M}$.
As we can see from the figure, \pnamemcex outperforms \prefname when the utilization increases (\ie $\frac{U}{M} > 0.2$). This is because our scheme allows security tasks to execute in parallel across cores and also allocate periods considering the schedulability constrains of all low priority tasks -- this results in a smaller response time and can find more tasksets that satisfy the designer specified bound. In contrast \prefname uses a greedy approach that minimizes the periods of higher priority tasks first without considering the global state. Also \prefname statically binds the security task to the core and hence suffers interference from the higher priority tasks assigned to that core -- this leads to lower acceptance ratios. 
For higher utilizations (\ie $\tfrac{U}{M} \geq 0.7$) \pnamemcex can find tasksets schedulable that can not be easily partitioned by using the \prefnamepartition scheme. The acceptance ratio of our method and the \prefnamepartition scheme is equal when $\tfrac{U}{M} < 0.7$. This is because, for lower utilizations some lower priority security tasks experience less interference due to longer periods and specific core assignment (recall we set $T_s = T_s^{max}$ for all security tasks). While we bind the RT tasks to cores (due to legacy compatibility), it does not affect the schedulability (\ie the acceptance ratio of \pnamemcex is higher when compared to the \prefnameglobal scheme). This is because, RT tasks are already schedulable when partitioned (\eg by assumption on taskset generation, see Section~\ref{sec:tc_generation}) and our analysis reduces the interference that RT tasks have on security ones.
For higher utilizations, the acceptance ratio drops for all the schemes since it is  not possible to satisfy all the constraints due to the high interference from RT and security tasks. We also highlight that while our approach results in better schedulability, \pnamemcex/\prefnamepartition (\ie where legacy RT tasks are partitioned to the cores) and \prefnameglobal  (\ie where all tasks can migrate) schemes are incomparable in general  (\eg there exists taskset that may be schedulable by task partitioning but not in global scheme where migration is allowed and vice-versa) 
-- we allow security tasks to migrate due to security requirements (\eg to achieve faster intrusion detection -- as we explain in the next experiments, see Fig.~\ref{fig:dist_diff_count}).

In the final set of experiments (Fig.~\ref{fig:dist_diff_count}) we compare the achievable periods (in terms of Euclidean distance) for our approach and the other schemes. The x-axis in the Fig.~\ref{fig:dist_diff_count} shows the normalized utilizations and the y-axis represents the average difference between the following period vectors: \ca between \pnamemcex  and \prefname (dashed line); \cb \pnamemcex  and other strategies (\eg \prefnameglobal and \prefnamepartition)  that do not consider period minimization 
(dotted marker) for dual and quad core setup. Higher distance values imply that the periods calculated by \pnamemcex are smaller (\ie leads to faster detection time) and our approach outperforms the other scheme. For low to medium utilizations (\eg $0.2 \leq U \leq 0.5$) \pnamemcex performs better when compared to \prefname. In situations with higher utilizations, the lesser availability of slack time results in \pnamemcex  and \prefname performing in a similar manner. Also, for higher utilizations \prefname is unable to find schedulable tasksets and hence there exist fewer data points. 

Our experiments also show that compared to \prefnameglobal and \prefnamepartition our approach finds smaller periods in most cases (Fig.~\ref{fig:dist_diff_count}).  This is expected since there is no period adaptation (\ie we set $T_s = T_s^{max}$ for those schemes). However it is important to note that \pnamemcex achieves better execution frequency (\ie smaller periods) without sacrificing schedulability as seen in Fig.~\ref{fig:sched_count}. That is, our semi-partitioned approach achieves better continuous monitoring when compared with both a fully-partitioned approach (\prefname, \prefnamepartition) and a global scheduling approach (\prefnameglobal) while providing the same or better schedulability.



\section{Discussion}



In this paper we do not design for any specific security tasks (the IDS system used is meant for demonstration purposes only) and allow designers to integrate their preferred techniques. Depending on the actual implementation of the security tasks some attack may not be detectable. For instance, the system may be vulnerable to zero-day attacks if the security tasks are not designed to detect unforeseen exploits or anomalous behaviors. There exists cases where security tasks may require some amount of system modifications and/or porting efforts -- say a timing behavior based security checking~\cite{hamad2018prediction,securecore,dragonbeam} may require the insertion of probing mechanisms inside the RT application tasks (or additional hardware) so that security tasks can validate their expected execution profiles.

\pnamemcex abstracts security tasks (and underlying monitoring events) and works in a proactive manner. However, designers may want to integrate security tasks that \textit{react}, based on anomalous behavior. For instance, let at time $t$, $j$-th job of task $\tau_s$ (\eg $\tau_s^j$) performs action $\mathtt{a}_0$ (\eg runtime of real-time tasks). Because of intrusions (or perhaps due to other system artifacts) in time $[t, t+T_s]$ ($T_s$ is the period of $\tau_s$), job $\tau_s^{j+1}$ finds that $\mathtt{a}_0$ is not behaving as expected. Therefore $\tau_s^{j+1}$ may perform both actions, $\mathtt{a}_0$ and $\mathtt{a}_1$ (say that checks the list of system calls, to see if any undesired calls are executed). One way to support such a feature is to consider the dependency (\ie $\mathtt{a}_1$ depends on $\mathtt{a}_0$ in this case) between security checks (\eg sub-tasks). We intend to extend our framework considering dependency between security tasks.
\section{Related Work} \label{sec:rel_work}


\subsubsection*{RT Scheduling and Period Optimization}
Although not in the context of RT security, the scheduling approaches present in this paper  can be considered as a special case of prior work~\cite{linux_push_pull} where each task can bind to any arbitrary number of available cores. For a given period, this prior analysis~\cite{linux_push_pull} is pessimistic for the model considered by \pnamemcex (\ie RT tasks are partitioned and security tasks can migrate on any core) in a sense that it over-approximates carry-in interference from the tasks bound to single cores (\eg RT tasks) and hence results in lower schedulability (\eg identical to the \prefnameglobal scheme in Fig.~\ref{fig:sched_count}). 
Researchers also propose various semi-partitioned scheduling strategies for fixed-priority RTS~\cite{kato2009semi,lakshmanan2009partitioned}. However, this work  primarily focuses on improving schedulability (\eg by allowing highest priority task to migrate) and they are not designed for security requirements in consideration (\eg minimizing periods and executing security tasks with fewer interruption for faster anomaly detection). There exists other work~\cite{delay_period} in which
the authors statically assign the  periods for multiple independent control tasks considering control delay as a cost metric. Davare \etal~\cite{davare2007period_can} propose to assign task and message periods as well as satisfy end-to-end latency
constraints
for distributed automotive systems. While the previous work focus on optimizing period of \textit{all} the tasks in the system for a single core setup, our goal is to ensure security without violating timing constraints of the RT tasks in a multicore platform.


\subsubsection*{Security Solutions for RTS}


In recent years researchers proposed various mechanisms to provide security guarantees into legacy and non-legacy RTS (both single and multicore platforms) in several directions, \viz integration of security mechanisms~\cite{mhasan_rtss16,mhasan_ecrts17,mhasan_date18}, authenticating/encrypting communication channels~\cite{lesi2017network,lesi2017security,xie2007improving,xie2005dynamic,lin2009static,jiang2013optimization}, side-channel defence techniques~\cite{sg1,sg2,sibin_RT_security_journal,taskshuffler,volp_TT_randomization} as well as hardware/software-based frameworks~\cite{mohan_s3a,securecore,securecore_memory,slack_cornell,securecore_syscal,mhasan_resecure16,mhasan_resecure_iccps}. 


Perhaps the closest line of research is \prefname~\cite{mhasan_date18} 
where authors proposed to statically partition security tasks to the cores and used an optimization-based solution to obtain the periods. While this approach does not have the overhead of context switches across cores, as we observed from our experiments (Section \ref{sec:evaluation}), that scheme results in a poor acceptance ratio for larger utilizations, and suffers interference from other high priority tasks leading to slower detection of intrusions (\ie less effective). The problem of integrating security for single core RTS is addressed in prior research~\cite{mhasan_rtss16} where authors used hierarchical scheduling~\cite{server_ab_uk} 
and proposed to execute security tasks with a low priority server. This approach is also extended to a multi-mode framework~\cite{mhasan_ecrts17} that allows security tasks to execute in different modes (\ie passive monitoring with lowest priority as well as exhaustive checking with higher priority). These server-based approaches, however, may require additional porting efforts for legacy systems. 


There exists recent work~\cite{lesi2017network,lesi2017security} to secure cyber-physical systems from man-in-the-middle attacks by enabling authentication mechanisms and timing-aware network resource scheduling. There has also been work~\cite{xie2005dynamic,xie2007improving,lin2009static} where authors proposed to add protective security mechanisms into RTS and considered periodic task scheduling where each task requires a security service whose overhead varies according to the required level of service. 
The problem of designing secure multi-mode RTS have also been addressed in prior work~\cite{jiang2013optimization} under dynamic-priority scheduling. In contrast, we consider a multicore fixed-priority scheduling mechanism where security tasks are executed periodically, across cores, while meeting real-time requirements. The above mentioned work are designed for single core platforms and it is not straightforward to retrofit those approaches for multicore legacy systems. 



In another direction, the issues related to information leakage through storage timing channels using shared architectural resources (\eg caches) is introduced in prior work~
\cite{sg1, sg2, sibin_RT_security_journal}. The key idea is to use a modified fixed-priority scheduling algorithm with a state cleanup mechanism to mitigate information leakage through shared resources. However, this leakage prevention comes at a cost of reduced schedulability. 
Researchers also proposed to limit inferability of deterministic RT schedulers by randomizing the task execution patterns. Yoon \etal~\cite{taskshuffler}  proposed a schedule obfuscation method for fixed-priority RM systems. A combined online/offline randomization scheme~\cite{volp_TT_randomization} is also proposed to reduce determinism for time-triggered (TT) systems where tasks are executed based on a pre-computed, offline, slot-based schedule. We highlight that all the aforementioned work either requires modification to the scheduler or RT task parameters, and is designed for single core systems only.


Unlike our approach that works at the scheduler-level, researchers also proposed hardware/software-based architectural solutions~\cite{mohan_s3a,onchip_fardin,securecore,securecore_memory,slack_cornell,securecore_syscal,mhasan_resecure16,mhasan_resecure_iccps} to improve the security posture of future RTS. Those solutions require system-level modifications and are not suitable for legacy systems. To our knowledge this is the first work that aims to achieve continuous monitoring for multicore-based legacy RT platforms.
\section{Conclusion}
\label{sec:conclusion}


Threats to safety-critical RTS are growing and there is a need for developing layered defense mechanisms to secure such critical systems. We present algorithms to integrate continuous security monitoring for legacy multicore-based RTS. By using our framework, systems engineers can improve the security posture of RTS. This additional security guarantee also enhances safety -- which is the main goal for such systems.

\bibliographystyle{IEEEtran}

\bibliography{references_short}

\end{document}